\documentclass{llncs}

\usepackage{
  afterpage,   
  amsmath,     
  mathtools,   
  amssymb,     
  amsfonts,    
  pdfsync,      
  url,          
  xspace,       
  enumitem,      
  array			
  }

\usepackage{stmaryrd} 

\usepackage{centernot}
\usepackage{hyperref}
\usepackage{cleveref}
\crefname{figure}{Fig.}{Figs.}
\crefname{definition}{Def.}{Defs.}
\crefname{defn}{Def.}{Defs.}
\crefname{equation}{Eq.}{Eqs.}
\crefname{theorem}{Thm.}{Thms.}
\crefname{lemma}{Lem.}{Lems.}
\crefname{lem}{Lem.}{Lems.}
\crefname{corollary}{Cor.}{Cors.}
\crefname{proposition}{Prop.}{Props.}
\crefname{enumi}{Item}{Items}
\crefname{exmp}{E.g.}{E.g.}
\crefname{section}{Sec.}{Secs.}
\crefname{table}{Table}{Tables}
\crefname{appendix}{Appendix}{Apps.}
\crefname{lstlisting}{List.}{List.}
\Crefname{lstlisting}{List.}{List.}

\newtheorem{cor}{Corollary}

\usepackage{comment}
\includecomment{comment}
\specialcomment{karoliina}{\begingroup\color{blue}\sffamily\small}{\endgroup}
\specialcomment{antonis}{\begingroup\sffamily\color{red}\small}{\endgroup}
\specialcomment{adrian}{\begingroup\sffamily\color{cyan}\small}{\endgroup}

\newcommand{\qedhere}{\qed}

\usepackage{centernot}
\usepackage{hyperref}
\usepackage{cleveref}
\crefname{figure}{Fig.}{Figs.}
\crefname{definition}{Def.}{Defs.}
\crefname{defn}{Def.}{Defs.}
\crefname{equation}{Eq.}{Eqs.}
\crefname{theorem}{Thm.}{Thms.}
\crefname{thm}{Thm.}{Thms.}
\crefname{lemma}{Lem.}{Lems.}
\crefname{lem}{Lem.}{Lems.}
\crefname{corollary}{Cor.}{Cors.}
\crefname{cor}{Cor.}{Cors.}
\crefname{proposition}{Prop.}{Props.}
\crefname{prop}{Prop.}{Props.}
\crefname{enumi}{Item}{Items}
\crefname{example}{Example}{Examples}
\crefname{section}{Sec.}{Secs.}
\crefname{table}{Table}{Tables}
\crefname{appendix}{Appendix}{Apps.}
\crefname{lstlisting}{List.}{List.}
\Crefname{lstlisting}{List.}{List.}

\usepackage{tikz}  
\usetikzlibrary{matrix,shapes,arrows}
\tikzstyle{block} = [rectangle, draw, fill=white!20, text width=4em, text centered, minimum height=3em]
\tikzstyle{line} = [draw, very thick, color=black!90, -latex']

\usepackage[inference]{semantic} 

\usepackage{pic}  
\usepackage{monitor} 
\usepackage{logic}  
\usepackage{shorthand}  
\usepackage{mathpartir}  

\newcommand{\mveqinf}{\ensuremath{\simeq_\omega}}

\newcommand{\citeMac}[1]{\cite{#1}}

\newcommand{\dcup}{\sqcup}

\renewcommand{\hSemL}[1]{\hSem{#1}}

\newcommand{\yesac}{\heartsuit}
\newcommand{\noac}{\spadesuit}

\usepackage{scalerel}
\DeclareMathOperator*{\bigparalC}{\scalerel*{\paralC}{\sum}}

\newcommand{\qedmaybe}{}

\usepackage[noadjust]{cite}

\begin{document}

\title{The Cost of Monitoring Alone\thanks{This research was partially supported by the  projects  ``TheoFoMon: Theoretical Foundations for Monitorability'' (grant number: 
			163406-051%
	) and ``Epistemic Logic for Distributed Runtime Monitoring'' (grant number: 
		184940-051%
) of the Icelandic Research Fund, by the BMBF project ``Aramis II'' (project number: 
	01IS160253%
) and the EPSRC project ``Solving parity games in theory and practice'' (project number:
	EP/P020909/1
	).}}

\author{Luca Aceto\inst{1,2}\orcidID{0000-0002-2197-3018}  \and 
		Antonis Achilleos\inst{2}\orcidID{0000-0002-1314-333X} \and
		Adrian Francalanza\inst{3}\orcidID{0000-0003-3829-7391} \and
		Anna Ing\'{o}lfsd\'{o}ttir\inst{2}\orcidID{0000-0001-8362-3075} \and
		Karoliina Lehtinen\inst{4}\orcidID{0000-0003-1171-8790}
	}

\institute{
		Gran Sasso Science Institute
    \email{luca.aceto@gssi.it}
	\and 
		Reykjavik University
	\email{\{luca,antonios,annai\}@ru.is}  
	\and
		University of Malta
	\email{adrian.francalanza@um.edu.mt}    
	\and
		University of Liverpool
	\email{karoliina.lehtinen@liverpool.ac.uk}     
}

\maketitle

\begin{abstract}
%
We compare the succinctness of two monitoring systems for properties of infinite traces, namely parallel and regular monitors.
Although a parallel monitor can be turned into an equivalent regular monitor, the cost of this transformation is a double-exponential blowup in the syntactic size of the monitors, and a triple-exponential blowup when the goal is  a deterministic monitor. 
We show that these bounds are tight and that they also hold for translations between  corresponding  fragments of 
Hennessy-Milner logic with recursion
over infinite traces. 
\end{abstract}

\section{Introduction}

Runtime Verification is a lightweight verification technique where a computational entity that we call a monitor is used to observe a system run in order to verify a given property. That property, which we choose to formalize in 
Hennessy-Milner logic with recursion (\recHML) \cite{Larsen:90:HMLRec}, can be a potential property of either the system  \cite{FraAI:17:FMSD,AceAFI:17:FSTTCS}, or of the current system run, encoded as a trace of events \cite{AcetoAFIL19} --- see also, for example, \cite{PnueliZaks:06:FM,Falcone2012,bauer2011runtime} for earlier work on the monitoring of trace properties, mainly formalized on LTL.

To address the  case of verifying trace properties, the authors introduced in \cite{AcetoAFIL19} a class of monitors that can generate multiple parallel components that analyse the same system trace. These were called parallel monitors. When some of them reach a verdict, they can combine these verdicts into one.
In the same paper, it was determined that this monitoring system has the same monitoring power as its restriction to a single monitoring component, as it was used in \cite{FraAI:17:FMSD,AceAFI:17:FSTTCS}, called regular monitors. 
However, the cost of the translation from the more general monitoring system to this fragment, as 
given in \cite{AcetoAFIL19}, is doubly exponential with respect to the syntactic size of the monitors.
Furthermore, if the goal is a deterministic regular monitor \cite{AceAFI:2017:CIAA,determinization}, then the resulting monitor is quadruply-exponentially larger than the original, parallel one, in \cite{AcetoAFIL19}.

In this paper, we show that the double-exponential cost for translating from parallel to equivalent 
regular monitors is tight. Furthermore, we improve the translation cost from parallel monitors to equivalent deterministic monitors to a triple exponential, and we show that this bound is tight.
We define monitor equivalence in two ways, the first one stricter than the second. For the first definition, two monitors are equivalent when they reach the same verdicts for the same finite traces, while for the second one it suffices to reach the same verdicts for the same infinite traces.
We prove the upper bounds for a transformation that gives monitors that are equivalent with respect to the stricter definition, while we prove the lower bounds with respect to transformations that satisfy the 
coarser 
definition. Therefore, our bounds hold for both definitions of monitor equivalence.
This treatment allows us to derive stronger results, which yield similar bounds for the case of logical formulae, as well.

In \cite{AcetoAFIL19}, we show that, when interpreted over traces, \maxHML, the fragment of \recHML that does not use least fixed points, is equivalent to the syntactically 
smaller safety fragment \SHML.
That is, every \maxHML formula can be translated to a logically equivalent \SHML formula.
Similarly to the aforementioned translation of monitors, this translation of formulae results in a formula that 
that is syntactically at most doubly-exponentially larger than the original formula.
We show that this upper bound is tight.

The first four authors have worked on the complexity of monitor transformations before in \cite{AceAFI:2017:CIAA,determinization}, where the cost of determinizing monitors is examined.
Similarly to \cite{AceAFI:2017:CIAA,determinization}, in \cite{AcetoAFIL19}, but also in this paper, we use results and techniques from Automata Theory 
and specifically 
about alternating automata \cite{alternation,FJY90}. 

In \Cref{sec:preliminaries}, we introduce the necessary background on monitors and \recHML on infinite traces, as these were used in \cite{AcetoAFIL19}.
In \Cref{sec:transformations}, we describe the monitor translations that we mentioned above, and we provide upper bounds for these, which we prove to be tight in \Cref{sec:lower}.
In \Cref{sec:logical-consequences}, we extrapolate these bounds to the case where we translate logical formulae, from \maxHML to \SHML.
In \Cref{sec:conclusions}, we conclude the paper.
Omitted proofs can be found in the appendix.




\section{Preliminaries}
\label{sec:preliminaries}



Monitors are expected to monitor for a specification, which, in our case, is written in \recHML.
We use the linear-time interpretation of the logic \recHML, as it was given in \cite{AcetoAFIL19}.
%
According to that interpretation, formulae are interpreted over
 infinite \emph{traces}.


\subsection{The model and the logic} We assume a finite set of actions
\label{ssec:model-logic}
$\act,\actb,\ldots \in \Act$
with distinguished silent action $\actt$.
%
We also assume that $\tau\not\in \Act$ and that $\acttt \in \Act\cup\sset{\tau}$, and refer to the actions in $\Act$ as \emph{visible} actions (as opposed to the silent action $\tau$).
The metavariables $\tV,\tVV \in \Trc=\Act^\omega$ range over (infinite) sequences of visible actions, 
which
abstractly represent system runs.
We also use the metavariable $T \subseteq \Trc$ to range over \emph{sets of traces}.
%
We often need to refer to \emph{finite traces}, denoted as $\ftV,\ftVV \in \Act^\ast$, to represent objects such as a finite prefix of a system run, or to traces that may be finite or infinite (\emph{finfinite traces}, as they were called in \cite{AcetoAFIL19}), denoted as $\fftV,\fftVV \in \Act^\ast \cup \Act^\omega$.
A trace (\resp finite trace, \resp finfinite trace) with action \act at its head is denoted as $\act\tV$ (\resp $\act\ftV$, \resp $\act\fftV$). Similarly a trace with a prefix \ftV\ is 
written
$\ftV\tV$.

\begin{figure}[!h]
  \textbf{Syntax}
  \begin{align*}
      \hV,\hVV \in \recHML &\bnfdef  \hTru  
      &
                               &\bnfsepp  \hFls &
             & \bnfsepp \hOr{\hV\,}{\,\hVV}  &
             & \bnfsepp \hAnd{\hV\,}{\,\hVV}  
             \\
             &~~~\bnfsepp \hSuf{\act}{\hV} &
              &\bnfsepp \hNec{\act}{\hV} &
              & \bnfsepp \hMinX{\hV} &
              & \bnfsepp \hMaxX{\hV} &
              & \bnfsepp\; X 
    \end{align*}
  \textbf{Linear-Time Semantics}
	\begin{align*}
      \hmeaning{\hTru,\rho}  &\deftxt \Trc
      &
      \hmeaning{\hFls,\rho}  &\deftxt \emptyset\\
      \hmeaning{\hAnd{\hV_1}{\hV_2},\rho} & \deftxt \hmeaning{\hV_1,\rho} \cap \hmeaning{\hV_2,\rho}
       &
      \hmeaning{\hOr{\hV_1}{\hV_2},\rho} & \deftxt \hmeaning{\hV_1,\rho} \cup \hmeaning{\hV_2,\rho}\\
      \hmeaning{\hNec{\act}{\hV},\rho}  &  
                                          \deftxt \sset{\actb \tV \;|\;
                                          \actb \neq \act
                                          \text{ or }
                                          \tV \in \hmeaning{\hV,\rho}
                                          \!}
                                           &
      \hmeaning{\hSuf{\act}{\hV},\rho}  &  
                                          \deftxt \sset{\act\tV \;|\;
                                          \tV \in \hmeaning{\hV,\rho}
                                          \!}
                                             \\
      \!\hmeaning{\hMin{\!\hVarX}{\hV},\rho} & 
      	\deftxt 
      	\bigcap \sset{T \;|\;  \hmeaning{\hV,\rho[\hVarX\mapsto T]} \subseteq T}
                                          \\
      \!\hmeaning{\hMax{\!\hVarX}{\hV},\rho} & 
              \deftxt \bigcup \sset{T \;|\;  T \subseteq \hmeaning{\hV,\rho[\hVarX\mapsto T]}}
         &
      \hmeaning{\hVarX,\rho} & \deftxt \rho(\hVarX)
       \end{align*}
  \caption{\recHML Syntax and Linear-Time Semantics}
  \label{fig:recHML}
\end{figure}
%
%
%
%
The logic \recHML  \cite{Larsen:90:HMLRec,AceILS:2007}  assumes a countable  set $\LVars$ (with $\hVarX \in \LVars$) of logical variables, 
and is defined as  the set of \emph{closed} formulae 
generated by the grammar of \Cref{fig:recHML}. Apart from the standard constructs for truth, falsehood, conjunction and disjunction, the logic is equipped with possibility and necessity modal operators labelled by visible actions, together with recursive formulae expressing least or greatest fixpoints; formulae \hMinX{\hV} and  \hMaxX{\hV} 
bind free instances of the logical variable \hVarX in \hV, inducing the usual notions of open/closed formulae and  formula
equality up to alpha-conversion.

We interpret \UHML\ formulae over traces, 
using an interpretation function $\hSemL{-}$ that maps formulae to sets of traces, 
relative to an environment $\rho :\LVars \to 2^{\Trc}$, which intuitively assigns to each variable $\hVarX$ the set of traces that are assumed to satisfy it, 
as defined in \Cref{fig:recHML}.
The semantics of a closed formula $\hV$ is independent of the environment $\rho$ and 
is simply written $\hSemL{\hV}$.
Intuitively, $\hSemL{\hV}$ denotes the set of traces satisfying $\hV$.
For a formula $\hV$,
we use  $l(\hV)$ to denote the length of $\hV$ as a string of symbols.

\subsection{Two monitoring systems}
\label{sec:monitors}
\label{ssec:monitors}


\begin{figure}[t]
  \textbf{Syntax}
  \begin{align*}
    \mV,\mVV\in\Mon\ &\bnfdef\  \vV
    && \bnfsep~ \prf{\acta}{\mV}
    && \bnfsep~ \ch{\mV}{\mVV}
    && \bnfsep~ \rec{x}{\mV}
    &&\bnfsep~ x 
     &&
     \bnfsep~ \mV \paralC \mVV  && \bnfsep~ \mV \paralD \mVV \\
    \vV
    \in \Verd & \bnfdef\  \no
    && \bnfsep~ \yes
    &&  \bnfsep~ \stp
  \end{align*}
  \textbf{Dynamics}
   \begin{mathpar}
   	\texttt{Regular monitor rules:}
   	\hrulefill
   	\\\\
 \inference[\rtit{Act}]{
 }{\prf{\acta}{\mV}  \traSS{\acta} \mV}
 \and
 %
\inference[\rtit{RecF}]{}{\rec{x}{\mV} \traSS{\tau} \mV}
\and 
\inference[\rtit{RecB}]{}{x \traSS{\tau} 
	p_x
}
%
\\
   \inference[\rtit{SelL}]{\mV \traSS{\actu} \mV'}{\ch{\mV}{\mVV}  \traSS{\actu} \mV'}
   \and
  \inference[\rtit{SelR}]{\mVV \traSS{\actu} \mVV'}{\ch{\mV}{\mVV}  \traSS{\actu} \mVV'}
  \and
  \inference[\rtit{Ver}]{
  }{\vV \traSS{\acta} \vV}
\\\\
\hrulefill
\texttt{Parallel tracing rules:}
\hrulefill
\\\\ 
\inference[\rtit{Par}]{\mV\traSS{\acta}\mV' & \mVV\traSS{\acta}\mVV'}
{\mV\paralG\mVV\traSS{\acta}\mV'\paralG\mVV'}
\quad
\inference[\rtit{TauL}]{\mV\traSS{\tau}\mV'}
{\mV\paralG\mVV\traSS{\tau}\mV'\paralG\mVV}
\quad
 \\\\
 \texttt{Parallel evaluation rules:}
 \hrulefill
 \\\\
\inference[\rtit{VrE}]{}{\stp\paralG\stp \traS{\tau} \stp}
\and 
\inference[\rtit{VrC1}]{}{{{\yes\paralC\mV} \traS{\tau} {\mV}}}
\and
\inference[\rtit{VrC2}]{}{{{\no\paralC\mV} \traS{\tau} {\no}}}
 \\\\
\inference[\rtit{VrD1}]{}{{{\no\paralD\mV} \traS{\tau} {\mV}}}
\and
\inference[\rtit{VrD2}]{}{{{\yes\paralD\mV} \traS{\tau} {\yes}}}
\end{mathpar}
\hrule
  \caption{Monitor Syntax and Semantics}
  \label{fig:monit-instr}
\end{figure}


%

We now present
two monitoring systems, parallel and regular monitors, that were introduced in \cite{FraAI:17:FMSD,AceAFI:17:FSTTCS,AcetoAFIL19}.
A monitoring system is a \emph{Labelled Transition System (LTS)} based on $\Act$, the set of actions, that is comprised of the monitor states, or monitors, and a transition relation.
%
%
%
%
The set of monitor states, 
$\Mon$, and the monitor transition relation, $\reduc \subseteq (\Mon\times(\Act\cup\sset{\tau})\times \Mon)$, are defined in \Cref{fig:monit-instr}. 
There and elsewhere, $\paralG$ ranges over both parallel operators $\paralD$ and $\paralC$.
When discussing a monitor with free variables (an \emph{open monitor}) $\mV$, we assume 
it is part 
of a larger monitor $\mV'$ without free variables (a \emph{closed monitor}), where every variable $x$ appears at most once in a recursive operator.
Therefore, we assume an injective mapping from each monitor variable $x$ to a unique monitor $p_x$, of the form $\rec x \mV$ that is a submonitor of $\mV'$.

The suggestive notation $\mV \traS{\acttt} \mVV$ denotes $(\mV,\acttt,\mVV) \in \reduc$; we also write $\mV \traSN{\acttt}$ to denote $\neg(\exists \mVV. \;\mV\traS{\acttt}\mVV)$.
We employ the usual notation for weak transitions and write $\mV \wreduc \mVV$ in lieu of $\mV (\traS{\tau})^{\ast} \mVV$ and $\mV \wtraS{\acttt} \mVV$ for
$\mV \wreduc\cdot\traS{\acttt}\cdot\wreduc \mVV$.
We
write  sequences of transitions
$\mV\wtra{\act_1}\cdots\wtra{\act_k} \mV_k$ as $\mV \wtraS{\ftr} \mV_k$,
where $\ftr = \act_1\cdots\act_k$.
The monitoring system of \emph{parallel monitors} is defined using the full syntax and all the rules from \Cref{fig:monit-instr}; \emph{regular monitors} are parallel monitors that do not use the parallel operators $\paralC$ and $\paralD$.
Regular monitors were defined and used 
already in \citeMac{AceAFI:17:FSTTCS} and \citeMac{FraAI:17:FMSD}, while parallel monitors were defined in \cite{AcetoAFIL19}.
We observe that the 
rules \rtit{RecF} and \rtit{RecB} are not the standard recursion rules from \citeMac{AceAFI:17:FSTTCS} and \citeMac{FraAI:17:FMSD}, but they are equivalent to those rules \cite{AcetoAFIL19,aceto2019adventures} and 
more convenient
for our arguments.


 %
A transition $\mV \traS{\acta} \mVV$ denotes that the monitor in state \mV can \emph{analyse} the (visible) action \acta and transition to state \mVV.
%
%
Monitors may reach any one of \emph{three} verdicts after analysing a finite trace: \emph{acceptance}, \yes, \emph{rejection}, \no, and the \emph{inconclusive} verdict \stp.
We highlight the transition rule for verdicts in \Cref{fig:monit-instr}, describing the fact that from a verdict state any action can be analysed by transitioning to the same state;  verdicts are thus \emph{irrevocable}.
Rule \rtit{Par} states that \emph{both} submonitors need to be able to analyse an external action \acta for their parallel composition to transition with that action.
The rules in \Cref{fig:monit-instr} also allow \actt-transitions for the reconfiguration of parallel compositions of monitors.
For instance, rules \rtit{VrC1} and \rtit{VrC2} describe the fact that, in conjunctive parallel compositions, whereas $\yes$  verdicts are uninfluential,  $\no$  verdicts supersede the verdicts of other monitors 
(\Cref{fig:monit-instr} omits  the obvious symmetric rules).
The dual applies for $\yes$  and $\no$  verdicts in a disjunctive parallel composition, as described by rules \rtit{VrD1} and \rtit{VrD2} (again, we omit the obvious symmetric rules).
Rule \rtit{VrE} applies to both forms of parallel composition and consolidates multiple inconclusive verdicts.
Finally, rules \rtit{TauL} and its omitted dual \rtit{TauR} 
are contextual rules for these monitor reconfiguration steps.
%
%

\begin{definition}[Acceptance and Rejection]  \label{def:acc-n-rej}
	We say that \emph{\mV rejects} (\resp \emph{accepts}) $\ftV\in \Act^*$ when 
	$\mV \wtraS{\ftV} \no$ 
	(\resp 
	$\mV \wtraS{\ftV} \yes$).
	We similarly say that $\mV$ rejects (\resp accepts) $\tV \in \Act^\omega$
	if 
	\mV rejects (\resp accepts) 
	some prefix of $\tV$.
	%
\end{definition}

%
Just like for formulae,
we use  $l(\mV)$ to denote the length of $\mV$ as a string of symbols.
%
In the sequel, for a finite nonempty set of indices $I$, we use 
$\sum_{i \in I}\mV_i$ to denote any combination of the monitors in
$\{\mV_i \mid i \in I\}$
using the operator $+$.
The notation is justified, because $+$ is commutative and associative with respect to the transitions that a resulting monitor can exhibit.
%
For each $j \in I$, $\mV_j$ is called a summand of $\sum_{i \in I}\mV_i$  (and the term $\sum_{i \in I}\mV_i$ is called a sum of $\mV_j$).
%
%
The regular monitors in \Cref{fig:monit-instr} have an important property, namely that their state space, \ie the set of reachable states, is finite (see \Cref{rem:infinitestate}).
On the other hand, parallel monitors can be infinite-state, but they are convenient when one 
synthesizes monitors. 
However, 
the two monitoring systems are equivalent (see \Cref{prop:extended-mon-to-reg-mon}).
%
%
%
For a monitor $\mV$, $\reach{\mV}$ is the set of monitor states reachable through a transition sequence from $\mV$.

\begin{lemma}[Verdict Persistence, \cite{FraAI:17:FMSD,AcetoAFIL19}]\label{lem:ver-persistence}
  \begin{math}
    \vV \wtraS{\ftV} \mV \text{ implies } \mV=\vV.
  \end{math}
  \qedmaybe 
\end{lemma}

\begin{lemma}\label{lem:subs2subs}
	Every submonitor of a closed regular monitor $\mV$ can only transition to submonitors of $\mV$.
	\qedmaybe
\end{lemma}

\begin{remark}\label{rem:infinitestate}
	An immediate consequence of \Cref{lem:subs2subs} is that regular monitors are finite-state. This is not the case for parallel monitors, in general.
	For example, consider parallel monitor $\mV_\actt = \rec x (x \paralC (a.\yes + b.\yes))$. We can see that there is a unique sequence of transitions that can be made from $\mV_\actt$:
	\begin{align*} 
	\mV_\actt &\traS{\actt} x \paralC (a.\yes + b.\yes)  
	\traS{\actt} \mV_\actt \paralC (a.\yes + b.\yes) \\ 
	&\traS{\actt} (x \paralC (a.\yes + b.\yes) ) \paralC (a.\yes + b.\yes) \\
		&\traS{\actt}  (\mV_\actt \paralC (a.\yes + b.\yes) ) \paralC (a.\yes + b.\yes) \reduc \cdots 
		\tag*{\qedd} 
	\end{align*}
\end{remark}

%
%


One basic requirement that we maintain on monitors is that they are not allowed to give conflicting verdicts for the same trace.

\begin{definition}[Monitor Consistency] \label{def:consitency-totality}
\quad
 A monitor $\mV$ is \emph{consistent} when there is \emph{no} finite trace $\ftr$ such that  $\mV \wtraS{\ftr} \yes$ and $\mV \wtraS{\ftr} \no$.
\end{definition}

We identify a useful monitor predicate that 
allows us to neatly decompose the 
behaviour 
of a parallel monitor 
in terms of its 
constituent 
sub-monitors.

\begin{definition}[Monitor Reactivity]
\label{def:reactive-mon}
We call a monitor $\mV$ \emph{reactive} when for every $\mVV \in \reach{\mV}$ and $\act \in \Act$, there is some $\mVV'$ such that $\mVV \wtraS{\acta} \mVV'$. 
\end{definition}

The following lemma 
states that parallel monitors behave as expected with respect to the acceptance and rejection of traces as long as the constituent submonitors are reactive.

\begin{lemma}[\cite{AcetoAFIL19}] \label{cor:parallel-and-monitors}
	For reactive $\mV_1$ and $\mV_2$:
	\begin{itemize}
		\item
		$\mV_1 \paralC \mV_2$ rejects $\tV$ if and only if either $\mV_1$ or $\mV_2$ rejects $\tV$.
		\item
		$\mV_1 \paralC \mV_2$ accepts $\tV$ if and only if both $\mV_1$ and $\mV_2$ accept $\tV$.

		\item
		$\mV_1 \paralD \mV_2$ rejects $\tV$ if and only if both $\mV_1$ and $\mV_2$ reject $\tV$.
		\item
		$\mV_1 \paralD \mV_2$ accepts $\tV$ if and only if either $\mV_1$ or $\mV_2$ accepts $\tV$.
		\qedmaybe
	\end{itemize}
\end{lemma}

The following example, which stems from 
\cite{AcetoAFIL19}, 
indicates why the assumption that $\mV_1$ and $\mV_2$ are reactive is needed in \Cref{cor:parallel-and-monitors}.

\begin{example} \label{ex:reactiveness}
 Assume that $\Act=\sset{a,b}$.
  The monitors $\ch{\prf{a}{\yes}}{\prf{b}{\no}}$ and $\rec{x}{(\ch{\prf{a}{x}}{\prf{b}{\yes}})}$
  are both reactive.
  The monitor $\mV = \prf{a}{\yes}\paralC\prf{b}{\no}$, however, is \emph{not} reactive.
 Since the submonitor $\prf{a}{\yes}$ can only transition with $a$, according to the rules of \Cref{fig:monit-instr}, $\mV$ cannot transition with any action that is not $a$.
 Similarly, as the submonitor $\prf{b}{\no}$ can only transition with $b$, $\mV$ cannot transition with any action that is not $b$.
 Thus, $\mV$ cannot transition to any monitor, and therefore it cannot reject or accept any trace.
\end{example}

In general, we are interested in reactive parallel monitors, and the parallel monitors that we use will have this property.

\subsection{Automata, Languages, Equivalence}
\label{ssec:automata-lang-eq}
In \cite{AcetoAFIL19}, we describe how 
to
transform a parallel monitor to a verdict equivalent regular one.
This transformation goes through alternating automata \cite{alternation,FJY90}.
For our purposes, we only need to define nondeterministic and deterministic automata.
%

\begin{definition}[Finite Automata]
  A nondeterminitic finite automaton (NFA) is a quintuple $A = (Q, \Act, q_0, \delta, F)$, where $Q$ is a finite set of states,
  $\Act$ is a finite alphabet (here it coincides with the set of actions), $q_0$ is the starting state, $F \subseteq Q$ is the set of accepting, or final states, and $\delta \subseteq Q\times \Act \times Q$ is the transition relation.
An NFA
 is deterministic (DFA) if $\delta$ is a function from $Q \times \Act$ to $Q$.
\end{definition}


%
  Given a state $q \in Q$ and a symbol $\act \in \Act$,
  $\delta$ returns a 
  set of possible states where the NFA can transition, and we typically use $q' \in \delta(q,\act)$ instead of $(q,\act,q') \in \delta$.
  We  extend the transition relation
  to
  $\delta^*: Q\times \Act^* \to 2^Q$, so that $\delta^*(q, \varepsilon) = 
  \{q\}$ 
  and
  $\delta^*(q,\act \ftr) = \bigcup \{ \delta^*(q',\ftr) \mid q' \in \delta(q,\act) \}$.
  We say that the automaton \emph{accepts} $\ftr \in \Act^*$ when $\delta^*(q_0,\ftr) \cap F \neq \emptyset$, and that
  it recognizes $L\subseteq \Act^*$ when $L$ is the set of strings accepted by the automaton.

\begin{definition}[Monitor Language Recognition]\label{def:mon-language}
A monitor $\mV$ recognizes positively (\resp negatively) a set of finite traces (\ie a language) $L \subseteq \Act^*$ when for every $\ftr \in \Act^*$, $\ftr \in L$ if and only if $\mV$ accepts (\resp rejects) $\ftr$. We call the set that $\mV$ recognizes positively (\resp negatively) $L_a(\mV)$ (\resp $L_r(\mV)$).
Similarly, we say that 
$\mV$ recognizes positively (\resp negatively) a set of infinite traces $L \subseteq \Act^\omega$ when for every $\tV \in \Act^\omega$, $\tV \in L$ if and only if $\mV$ accepts (\resp rejects) $\tV$. 
\end{definition}

%
%
%
%

Observe that, by \Cref{lem:ver-persistence}, 
$L_a(\mV)$ and $L_r(\mV)$ are 
closed 
under finite extensions.

\begin{lemma}\label{lem:Loo=La.Act}
The set of \emph{infinite} traces that 
is recognized positively (\resp negatively)
by $\mV$ is exactly $L_a(\mV)\cdot \Act^\omega$ (\resp $L_r(\mV)\cdot \Act^\omega$). 
\end{lemma}

\begin{proof}
	The lemma is a consequence of verdict persistence (\Cref{lem:ver-persistence}).
	\qed 
\end{proof}

To compare different monitors, we use a notion of monitor equivalence from \citeMac{AceAFI:2017:CIAA} that focusses on how monitors can reach verdicts.

\begin{definition}[Verdict Equivalence] \label{def:verdict-eq}
	Monitors \mV and \mVV are 
	\emph{verdict equivalent}, denoted as $\mV\mveq\mVV$, if
	$L_a(\mV) = L_a(\mVV)$ and $L_r(\mV) = L_r(\mVV)$.
\end{definition}


One may consider the notion of verdict equivalence, as defined in \Cref{def:verdict-eq}, to be too strict.
After all, verdict equivalence is defined with respect to finite traces, so if we want to turn a  parallel monitor into a regular or deterministic monitor, the resulting monitor not only needs to accept and reject the same infinite traces, but it is required to do so at the same time the original parallel monitor does.
However, one may prefer to have a smaller, but not tight monitor, if possible, as long as it accepts the same infinite traces.


\begin{definition}[$\omega$-Verdict Equivalence] \label{def:inf-verdict-eq}
	Monitors \mV and \mVV are 
	\emph{$\omega$-verdict equivalent}, denoted as $\mV\mveqinf\mVV$, if
$L_a(\mV) \cdot \Act^\omega = L_a(\mVV)  \cdot \Act^\omega$ and $L_r(\mV)  \cdot \Act^\omega = L_r(\mVV)  \cdot \Act^\omega$.
\end{definition}
From \Cref{lem:Loo=La.Act} we observe that verdict equivalence implies $\omega$-verdict equivalence. The converse does not hold, because $\no \mveqinf \sum\limits_{\act \in \Act}\act.\no$, but $\no \not\mveq \sum\limits_{\act \in \Act}\act.\no$.

%
%
%

\begin{definition}[\cite{determinization}]\label{def:determinism}
	A closed regular monitor $\mV$ is \emph{
		deterministic} iff
	every sum of at least two summands
	that appears in $\mV$
	is of the form
	$\sum_{\act \in A} \act.m_\act$,
	where $A \subseteq \Act$.
\end{definition}

\begin{example}
	The monitor \ch{\prf{a}{\prf{b}{\yes}}}{\prf{a}{\prf{a}}{\no}} is not 
	deterministic while the verdict equivalent monitor \prf{a}{(\ch{\prf{b}{\yes}}{\prf{a}{\no}})} is 
	deterministic.
\end{example}

%


\subsection{Synthesis}
\label{ssec:synthesis}

There is a tight connection between the logic from \Cref{ssec:model-logic} and the monitoring systems from \Cref{ssec:monitors}.
Ideally, we would want to be able to synthesize a monitor from any formula $\hV$, such that the monitor recognizes $\hSemL{\hV}$ positively and $\Trc \setminus \hSemL{\hV}$ negatively.
However, as shown in \cite{AcetoAFIL19}, neither goal is possible for all formulae.
Instead, we identify the following fragments of $\recHML$.

\begin{definition}[MAX and MIN Fragments of \ltmu] \label{def:minHML-maxHML}
	The greatest-fixed-point and least-fixed-point fragments of \UHML are, respectively, defined as:
	\begin{align*}
	\hV,\hVV \in \ltmuS
	&\bnfdef  \hTru 
	\bnfsepp  \hFls 
	\bnfsepp \hOr{\hV}{\hVV}  
	\bnfsepp \hAnd{\hV}{\hVV} 
	\bnfsepp \hSuf{\act}{\hV} 
	\bnfsepp \hNec{\act}{\hV} 
	 \bnfsepp \hMaxXF ~\text{ and}
	\\
	\hV,\hVV \in \ltmuC
	&\bnfdef  \hTru 
	\bnfsepp  \hFls 
	\bnfsepp \hOr{\hV}{\hVV}  
	\bnfsepp \hAnd{\hV}{\hVV} 
	\bnfsepp \hSuf{\act}{\hV} 
	\bnfsepp \hNec{\act}{\hV} 
	 \bnfsepp \hMinXF .
	\end{align*}
\end{definition}

\begin{definition}[Safety and co-Safety Fragments of \ltmu] \label{def:sHML-cHML}
	The safety and co-safety fragments of \UHML are, respectively, defined as:
	\begin{align*}
	\hV,\hVV \in \SHML
	&\bnfdef  \hTru &
	&\bnfsepp  \hFls &
	&\bnfsepp \hAnd{\hV\,}{\,\hVV} &
	&\bnfsepp \hNec{\act}{\hV} &
	& \bnfsepp \hMaxXF   ~\text{ and}
	\\
	\hV,\hVV \in \CHML
	&\bnfdef  \hTru &
	&\bnfsepp  \hFls &
	&\bnfsepp \hOr{\hV\,}{\,\hVV}  &
	&\bnfsepp \hSuf{\act}{\hV} &
	& \bnfsepp \hMinXF    .
	\end{align*}
\end{definition}

\begin{theorem}[Monitorability and Maximality, \cite{AcetoAFIL19}] \label{thm:monitorability-maximality}
	\begin{enumerate}
		\item 
			For every $\hV \in \ltmuS$ (\resp $\hV \in \ltmuC$), there is a reactive parallel monitor $\mV$, such that $l(\mV) = O(l(\hV))$ and $L_r(\mV)\cdot \Act^\omega = \Act^\omega\setminus \hSemL{\hV}$ (\resp $L_a(\mV)\cdot \Act^\omega = \hSemL{\hV}$).
		\item 
		For every reactive parallel monitor $\mV$, there are 
		$\hV\in \ltmuS $ and $  \hVV \in \ltmuC$, such that $l(\hV), l(\hVV) = O(l(\mV))$, $L_r(\mV)\cdot \Act^\omega = \Act^\omega\setminus\hSem{\hV}$, and $L_a(\mV)\cdot \Act^\omega = \hSem{\hVV}$.
		\item 
		For every $\hV \in \SHML$ (\resp $\hV \in \CHML$), there is a regular monitor $\mV$, such that $l(\mV) = O(l(\hV))$ and $L_r(\mV)\cdot \Act^\omega = \Act^\omega \setminus \hSemL{\hV}$ (\resp $L_a(\mV)\cdot \Act^\omega = \hSemL{\hV}$).
		\item 
		For every regular monitor $\mV$, there are 
		$\hV\in \SHML $ and $  \hVV \in \CHML$, such that $l(\hV), l(\hVV) = O(\mV)$, $L_r(\mV)\cdot \Act^\omega =\Act^\omega \setminus \hSemL{\hV}$, and $L_a(\mV)\cdot \Act^\omega = \hSemL{\hVV}$.
		\qedmaybe 
	\end{enumerate}
\end{theorem}

We say that a logical fragment is monitorable for a monitoring system, such as parallel or regular monitors, when for each of the fragment's formulae there is a monitor that detects exactly the satisfying or violating traces for  that formula.
One of the consequences of 
\Cref{thm:monitorability-maximality} 
is that the fragments defined in \Cref{def:minHML-maxHML,def:sHML-cHML} are semantically the largest monitorable fragments of \UHML 
for 
parallel and regular monitors, respectively. 
As we will see in \Cref{sec:transformations}, every parallel monitor has a verdict equivalent regular monitor (\Cref{prop:determinization,prop:better_determinization}), and therefore all formulae in \ltmuC and \ltmuS can be translated into equivalent \CHML and \SHML formulae respectively, as \Cref{thm:samelogics} later on demonstrates.
However, \Cref{thm:logic-lower-bounds} to follow shows that the cost of this translation is significant.

\section{Monitor Transformations: Upper Bounds}
\label{sec:transformations}

In this section we explain how to transform a 
parallel monitor into a regular or deterministic monitor, and what is the cost, in monitor size, of this transformation. 
The various relevant transformations, including some from \cite{determinization} and \cite{alternation,FJY90}, are summarized in \Cref{fig:transformations}, where each edge is labelled with the best-known worst-case upper bounds for the cost of the corresponding transformation in \Cref{fig:transformations} (AFA abbreviates alternating finite automaton \cite{alternation}). As we see in \cite{AceAFI:2017:CIAA,determinization} and in \Cref{sec:lower}, these bounds cannot be improved significantly. 
\begin{figure}
	\begin{center}
		\begin{tikzpicture}
		\node[scale=0.8](pmon) at (-2,1.3) {\begin{tabular}{c}
			parallel \\monitor
			\end{tabular} };
		\node[scale=0.8] (afa) at (-2,3)  {AFA};
		\node[scale=0.8](mon) at (2,1.3) {\begin{tabular}{c}
			regular \\monitor
			\end{tabular} };
		\node[scale=0.8] (nfa) at (2,3)  {NFA};
		\node[scale=0.8] (dfa) at (6,3)  {DFA};
		\node[scale=0.8] (dmon) at (6,1.3) {\begin{tabular}{c}
			deterministic \\monitor
			\end{tabular} };
		

		\path[->] (pmon) edge node[left,scale=0.8] {$O(n)$} (afa);
		\path[->] (afa) edge [bend left=10]  node[above,scale=0.8] {$O(2^n)$} (nfa);
		\path[->] (pmon) edge [bend left=10]  node[above,scale=0.8] {$2^{O(n \cdot 2^n)}$} (mon);
		\path[->] (mon) edge [bend right=10] node[above right,scale=0.8] {$O(n)$} (nfa);
		\path[->] (nfa) edge [bend left=10]  node[above,scale=0.8] {$O(2^n)$} (dfa);
		\path[->] (dfa) edge node[left,scale=0.8] {$2^{O(n)}$} (dmon);
		\path[->] (mon) edge [bend left=10] node[above,scale=0.8] {
			$2^{O(2^n)}$} (dmon);
		\path[->] (nfa) edge [bend right=10] node[above left,scale=0.8] {
			$2^{O(n\log n)}$} (mon);
		\end{tikzpicture}
	\end{center}
	\caption{Monitor Transformations and Costs}
	\label{fig:transformations}
\end{figure}
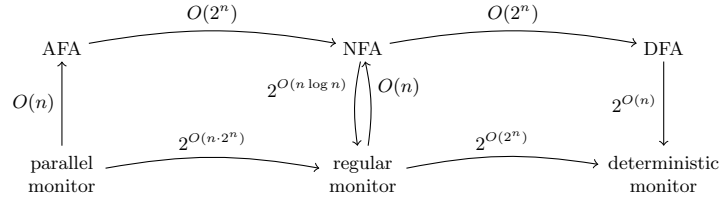
\begin{proposition}[\cite{AcetoAFIL19}]\label{prop:monitor2automaton}
    For every reactive 
    parallel 
    monitor $\mV$, there are an alternating automaton that
    recognizes $L_a(\mV)$
%
    and one that
    recognises $L_r(\mV)$, with $O(l(\mV))$ states.
  \end{proposition}
  \begin{proof}
  	The construction from \cite[Proposition 3.6]{AcetoAFIL19} gives an automaton that has the submonitors of $\mV$ as states. The automaton's transition function corresponds to the semantics of the monitor.
  \qed  \end{proof}

  \begin{cor}[Corollary 3.7 of \cite{AcetoAFIL19}]\label{cor:parallel2NFA}
    For every reactive  and closed parallel monitor $\mV$, there are an NFA that
    recognises $L_a(\mV)$
    and an NFA that
    recognises $L_r(\mV)$, and each has at most $2^{l(\mV)}$ states.
  \end{cor}


%



\begin{proposition}\label{prop:extended-mon-to-reg-mon}
For every reactive and closed 
parallel monitor \mV, there exists 
a verdict equivalent
regular monitor 
\mVV such that
$l(\mVV) = 2^{O\left(l(\mV)\cdot 2^{l(\mV)}\right)}$.
\end{proposition}
\begin{proof}
Let $A_\mV^a$ be an NFA for $L_a(\mV)$ with at most $2^{l(\mV)}$ states, and let $A_\mV^r$ be an NFA for $L_r(\mV)$ with at most $2^{l(\mV)}$ states, which exist by \Cref{cor:parallel2NFA}.
From these NFAs, we can construct regular monitors $\mV_R^a$ and $\mV_R^r$,
such that
$\mV_R^a$ recognizes $L_a(\mV)$ positively and $\mV_R^r$ recognizes $L_r(\mV)$ negatively, and $l(\mV_R^a), l(\mV_R^r) = 2^{O\left(l(\mV)\cdot 2^{l(\mV)}\right)}$ \cite[Theorem 2]{AceAFI:2017:CIAA}.
Therefore,
$\mV_R^a + \mV_R^r$ is regular and verdict equivalent to $\mV$, and $l(\mV_R^a + \mV_R^r) =  2^{O\left(l(\mV)\cdot 2^{l(\mV)}\right)}$.
\qed  \end{proof}

The constructions from \cite{AcetoAFIL19} include a determinization result, based on \cite{determinization}.


%
%
%

\begin{theorem}[Corollary 3 of \cite{AceAFI:2017:CIAA}]\label{thm:determ_from_determ}
	For every consistent  closed  regular monitor $\mV$, there is a 
	deterministic 
	monitor $\mVV$ such that
	$\mVV \mveq \mV$ and
	$l(\mVV) = 2^{2^{O(l(\mV))}}$.
	\qedmaybe
\end{theorem}

\begin{proposition}[Proposition 3.11 of \cite{AcetoAFIL19}]\label{prop:determinization}
For every consistent reactive and closed parallel monitor $\mV$, there is a verdict equivalent deterministic regular monitor $\mVV$ such that $l(\mVV) = 2^{2^{2^{O\left(l(\mV)\cdot 2^{l(\mV)}\right)}}}$.
\qedmaybe
\end{proposition}


However, the bound given 
by \Cref{prop:determinization}
for the  construction 
of deterministic regular monitors from parallel ones 
 is not optimal, as we 
 observe below.

\begin{proposition}\label{prop:better_determinization}
	For every consistent reactive and closed parallel monitor $\mV$, there is a 
deterministic 
	monitor $\mVV$ such that $\mVV \mveq \mV$ and $l(\mVV) = {2^{2^{2^{O\left(l(\mV)\right)}}}}$.
	\qedmaybe
\end{proposition}

In the following \Cref{sec:lower}, we see that the upper bounds of \Cref{prop:extended-mon-to-reg-mon,prop:better_determinization}
are tight, even for 
monitors that can only accept or reject, and even when the constructed regular or deterministic monitor is only required to be $\omega$-verdict equivalent to the starting one, and not necessarily verdict equivalent, to the original parallel monitor.
As we only need to focus on acceptance monitors, in the following we say that a monitor recognizes a language to mean that it recognizes the language positively.


\section{Lower Bounds}
\label{sec:lower}

We now prove that the transformations of 
\Cref{sec:transformations} are optimal, by establishing the corresponding lower bounds.
To this end, we introduce a family of suffix-closed languages $L_A^k \subseteq \{ 0, 1, e, \$, \# \}^*$. Each $L_A^k$ is a variation of a language introduced in \cite{alternation} to prove the $2^{2^{o(n)}}$ lower bound for the transformation from an alternating automaton to a deterministic one.
In this section, we only need to consider closed monitors, and as such, all monitors are assumed to be closed.

A variation of the language that was introduced in \cite{alternation} is the following:
\[
	L_V^k = \{ u \# w \# v \$ w \mid u, v \in \{ 0,1,\# \}^*,\ w \in \{0,1\}^{k} \}.
\]

An alternating automaton that recognizes $L_V^k$
can nondeterministically skip to the first occurrence of $w$ and then verify that, for every number $i$ between $0$ and $k-1$, the $i$'th bit matches the $i$'th bit after the $\$ $ symbol. This verification can be done using up to $O(k)$ states, to count the position $i$ of the bit that is checked. On the other hand, a DFA that recognizes $L_V^k$ must remember all possible candidates for $w$ that have appeared before $\$$, and hence requires $2^{2^k}$ states.
We can also conclude that any NFA for $L_V^k$ must have at least $2^k$ states, because a smaller NFA could be determinized to a smaller DFA.

\paragraph{A gap language}
For our purposes, we 
use  a similar family $L_A^k$ of suffix-closed
languages,
which are 
designed to be
recognized
by small parallel monitors,
but such that each regular monitor recognizing $L_A^k$ must be ``large''. 
%
%
We fix two numbers $l,k\geq 0$, such that $k = 2^l$.
First, we observe that we can encode every string $w \in \{0,1\}^k$  as a string $a_1 \act_1 a_2 \act_2\cdots a_k \act_k \in \{0,1\}^{(l
	+1)\cdot k}$, where $a_1a_2\cdots a_k$ is a permutation of $\{0,1\}^l$
and, for all $i$,  $\act_i \in \{0,1\}$. Then, $a_i\act_i$ gives the information that, for $j$ being the number with binary representation $a_i$, the $j$'th position of $w$ holds bit $\act_i$.
Let 
%
\[
	W = \left\{ a_1 \act_1 
	\cdots a_k \act_k \in \{0,1\}^{(l+1)\cdot k} \bigm|
	\begin{tabular}{l}
		for all $1 \leq i \leq k,
		~ \act_i \in \{0,1\}$ and \\ 
		$a_i \in \{0,1\}^l$,  and   $a_1
		\cdots a_k \in \{0,1\}^{l}! $
	\end{tabular}
	\right\}.
 \]
Let $w,w' \in W$, where $w = a_1 \acta_1 a_2 \acta_2\cdots a_k \acta_k $ and $w' = b_1 \actb_1 b_2 \actb_2\cdots b_k \actb_k $, and for $1 \leq i \leq k$, $a_i,b_i \in \{0,1\}^l$ and $\acta_i, \actb_i \in \{0,1\}$.
We define $w \equiv w'$ to mean that for every $ i,j \leq k$, $a_i = b_j$ implies $\acta_i = \actb_j$.
Let $a = \act_0 \ldots \act_{l-1} \in \{0,1\}^l$; then, $enc(a) = bin(0)\act_0 \ldots bin(l-1)\act_{l-1}$ is the ordered encoding of $a$, where $bin(i)$ is the binary encoding of $i$.
Then, 
$w \in W$ is called an encoding of $a$ if $w \equiv enc(a)$.



%
Let $\Sigma = \{ 0,1,\# \}$ and $\Sigma_\$ = \Sigma \cup \{\$\}$. Then,
\[
	L_A^k = \{ u \# w \# v \$ u'\#w'\#\$v' \mid u,v' \in \Sigma_\$^*, u',v \in \Sigma^* ,\ w, w' \in W,\
	\text{ and } w \equiv w' \}.
\]

In other words, a finite trace is in $L_A^k$ exactly when 
it has a substring of the form $\# w \# v \$ u'\#w'\#\$ $, where $w$ and $w'$ are encodings of the same string and there is only one 
 $\$ $ between them. Intuitively, $\#$ is there to delimit bit-strings that may be elements of $W$, and $\$ $ delimits sequences of such bit-strings.
 So, the language asks if there are two such consecutive sequences where the last bit-string of the second sequence comes from $W$ and matches an element from the first sequence. We observe that $L_A^k$ is suffix-closed.
 
\begin{lemma}\label{lem:lang1-fin-inf-trc}
	$\ftV \in L_A^k$
	if and only if
	$\forall \tV.~\ftV\tV
	\in L_A^k \cdot \Sigma_\$^\omega$.
	\qedmaybe 
\end{lemma}



\paragraph{Conventions}
For the conclusions of \Cref{cor:parallel-and-monitors} to hold, monitors need to be reactive. However, a reactive monitor can have a larger syntactic description than an equivalent non-reactive one, \eg $\acta.\yes$ \emph{vs.}~$\acta.\yes + \actb.\stp + \actg.\stp$, when $\Act = \{\acta, \actb, \actg \}$. This last monitor is also verdict equivalent to $\acta.\yes + \stp$.
In 
what follows,
for brevity and clarity, whenever we write a sum $s$ of a monitor of the form $\act.\mV$, we will mean $s+\stp$, which is reactive, so it can be safely used with a parallel operator, and is verdict equivalent to $s$.
We use set-notation for monitors: for $A \subseteq \Act$, $A.\mV$ stands for $\sum_{\act \in A}\act.\mV$ (or $\sum_{\act \in A}\act.\mV + \stp$ under the above convention). Furthermore, 
we define 
$\{0,1\}^0.m = m$ and $\{0,1\}^{i+1}.m = 0.\{0,1\}^{i}.m + 1.\{0,1\}^{i}.m$.
Notice that $l(\{0,1\}^{i}.m) = 2^i \cdot l(m) + 5 \cdot (2^{i} - 1)$. 
We can also similarly define $T.m$ for $T \subseteq \{0,1\}^i$.

\paragraph{Auxiliary monitors}
We start by defining auxiliary monitors. Given a (closed) monitor $\mV$, let
\begin{align*}
	skip_\#(\mV) &:= \rec x  ((0.x + 1.x + \#.x) \paralD \#.\mV),
\\
	next_\#(\mV) &:= \rec x  (0.x + 1.x + \#.\mV),
\\
	next_\$(\mV) &:= \rec x  (0.x + 1.x + \#.x + \$.\mV),  \ \text{ and } \\
	skip\_last(\mV) &:= \rec x  (0.x + 1.x + \#.x
	+
	\#.(\mV \paralC \rec y (0.y + 1.y + \#.\$.\yes))
	).
\end{align*}
These monitors read the trace until they reach a certain symbol, and then they activate submonitor $\mV$. We can think that $skip_\#(\mV)$ nondeterministically skips to some occurrence of $\#$ that comes before the first occurrence of $\$ $; $next_\#(\mV)$ and $next_\$(\mV)$ respectively skip to the next occurrence of $\#$ and $\$ $; and $skip\_last(\mV)$ skips to the last occurrence of $\#$ before the next occurrence of $\#\$ $.


\begin{lemma}\label{lem:skipsharp}
	$skip_\#(\mV)$ accepts 
	$\fftV$ iff there are $\ftV$ and $\fftVV$, such that $\ftV \# \fftVV = \fftV$,  $\mV$ accepts $\fftVV$, and $\ftV \in \{0,1,\#\}^*$.
	\qedmaybe
\end{lemma}
%

The following lemmata are straightforward and explain how the remaining monitors defined above are used.

\begin{lemma}
	$next_\#(\mV)$ accepts 
	$\fftV$ iff there are $\ftV$ and $\fftVV$, such that 
	$\ftV\#\fftVV = \fftV$,  $\mV$ accepts $\fftVV$, and $\ftV \in \{0,1\}^*$.
	\qedmaybe
\end{lemma}

\begin{lemma}
	$next_\$(\mV)$ accepts 
	$\fftV$ iff there are $\ftV$ and $\fftVV$, such that 
	$\ftV \$ \fftVV = \fftV$, $\mV$ accepts $\fftVV$, and $\ftV \in \{0,1,\#\}^*$.
\qedmaybe
\end{lemma}

\begin{lemma}
	$skip\_last(\mV)$ accepts 
	$\fftV$ iff there are $\ftV,\ftVV$, and $\fftVV$, such that  $\ftV \# \ftVV \# \$ \fftVV = \fftV$, $\mV$ accepts $\ftVV \# \$ \fftVV$, $\ftVV \in \{0,1\}^*$, and $\ftV \in \{0,1,\#\}^*$.
	 \qedmaybe
\end{lemma}

The following monitors help us ensure that a bit-string from $ \{0,1\}^{(l+1) \cdot k}$ is actually a member of $W$.
Monitor $all$ ensures that all bit positions appear in the bit-string;
$no\_more(\ftV)$ assures us that the 
bit position
$\ftV$ does not appear any more in the remainder of the bit-string; and
 $unique$ 
 guarantees
 that each bit position appears at most once. Monitor $perm$ combines these monitors together.

\[
	all := {\paralC} \{	\rec x (\{0,1\}^{l +1}.x ~\paralD~ \ftV.
	\{0,1\}.
	\yes)	\mid  \ftV \in \{0,1\}^{l}	\};
\]
for $\ftV\in \{0,1\}^l$,
\[ no\_more(\ftV): = \rec x \left(\#.\yes
						+
						(\{0,1\}^{l} \setminus \{\ftV\}).\{0,1\}.x
	\right);  \]
\[
	unique :=
	\rec x \left(\#.\yes
	+
	\left(
				\{0,1\}^{l +1}.x ~\paralC~
					\sum_{\ftV \in \{0,1\}^{l}} \ftV.\{0,1\}. no\_more(\ftV)
	\right)\right);
\text{ and}
\]
\[ perm := all ~\paralC~ unique . \]

The purpose of $perm$ is to ensure that a certain block of bits before the appearance of the $\#$ symbol is a member of the set $W$:
it accepts $w\#$ exactly when $w$ is a sequence of blocks of bits with length exactly $l+1$ (by $unique$) and for every $a \in \{0,1\}^l$ there is some $\act \in \{0,1\}$ such that $a \act$ is one of these blocks (by $all$), and that for each such $a$ only one block is of the form $a \act'$ (by $unique$).

\begin{lemma}
	$perm$ accepts $\fftV$ iff $w\#$ is a prefix of $\fftV$, for some $w \in W$.
	\qedmaybe
\end{lemma}

\begin{lemma}
	$l(perm) = {O(k^2)} $.
	\qedmaybe
\end{lemma}

Given a block $\ftV$ of $l+1$ bits, monitor $find(\ftV)$ accepts a sequence of blocks of $l+1$ bits $w$ exactly when $\ftV$ is one of the blocks of $w$:
\[ find(\ftV):= \rec x (\ftV.\yes + (\{0,1\}^{l+1}\setminus \{\ftV\}).x) . \]

\begin{lemma}
	For $\ftV \in \{0,1\}^{l+1}$, $find(\ftV)$ accepts $\fftV$ if and only if there is some $\ftVV \in (\{0,1\}^{l+1})^*$, such that $\ftVV \ftV$ is a prefix of $\fftV$.
	\qedmaybe
\end{lemma}


For $\ftV \in \{0,1\}^{l+1}$, $match(\ftV)$ ensures that right before the second occurrence of $\$$, there is a $\# w \# $, where $w \in (\{0,1\}^{l+1})^+$ and $\ftV$ is 
a
$(l+1)$-bit block in $w$. 
\[
match(\ftV) :=
next_\$(skip\_last(
find(\ftV)
))
.
\]

\begin{lemma}
	For $\ftV \in \{0,1\}^{l+1}$, $match(\ftV)$ accepts $\fftV$ if and only if there are $\ftVV\$\ftVV'\# w \#\$ \fftVV = \fftV$, such that $\ftVV,\ftVV' \in \Sigma^*$, $w \in \{0,1\}^*$, and there is a prefix $w'\ftV$ of $w$, such that $w' \in (\{0,1\}^{l+1})^*$.
	\qedmaybe
\end{lemma}

\paragraph{Recognizing $L_A^k$ with a parallel monitor}
We can now define a parallel
acceptance-monitor of length $O(k^2)$ that recognizes $L_A^k$.
Monitor $matching$ ensures that every one of the consecutive blocks of $l+1$ bits that follow, also appears in the block of bits that appears right before the  occurrence of $\#\$ $ that follows the next $\$ $ (and that there is no other $\$ $ between these $\$ $ and $\#\$ $). Therefore, if what follows from the current position in the trace and what appears right before that occurrence of $\#\$ $ are elements $w,w'$ of $W$, $matching$ ensures that $w \equiv w$.
Then, $m_A^k$ nondeterministically chooses an occurrence of $\#$, it verifies that the block of bits that follows is an element $w$ of $W$ that ends with $\#$, and that what follows is of the form 
$v \$ u'\#w'\#\$v'$, where $ u',v \in \Sigma^* ,\ w' \in W,\
\text{ and } w \equiv w'$, which matches the description of $L_A^k$.
\begin{align*}
	matching :=& 	\rec x \sum_{\ftV \in \{0,1\}^{l +1}}
	\ftV.
	(
	x ~\paralC~
	match(\ftV)
	)\\
	m_A^k:=&
	skip_\#\left(
		perm ~\paralC~ 
		next_\$(skip\_last(perm))  ~\paralC~ 
matching
	\right)
\end{align*}

\begin{lemma}\label{lem:mA-good-mon-for-gap-det}
	$m_A^k$ recognizes $L_A^k$ and $l(m_A^k) = O(k^2)$.
\end{lemma}

\begin{proof}
The lemma follows from this section's previous lemmata and from counting the sizes of the various monitors that we have constructed.
\qed  \end{proof}

\begin{lemma}\label{cor:lower-for-det}
	If $m$ is a deterministic monitor that recognizes $L_A^k \cdot \Sigma_\$^\omega$, then $$|m| \geq \left(\big(2^{2^{k-1}}-1\big)!\right)^2 = 2^{\Omega\big(2^{2^{k-1} + k}\big)}.$$
	\qedmaybe 
\end{lemma}


\Cref{prp:lower_deterministic} gathers our conclusions about $L_A^k$.

\begin{theorem}\label{prp:lower_deterministic}
	For every $k>0$,
	$L_A^k $
	is recognized
	by an alternating automaton of $O(k^2)$ states and a parallel monitor of length $O(k^2)$, but by no DFA with $2^{o(2^k)}$ states and no deterministic monitor of length $2^{o(2^{2^{k-1} + k})}$.
%
$L_A^k \cdot \Sigma_\$^\omega $
is recognized
by 
a parallel monitor of length $O(k^2)$, but by no 
deterministic monitor of length $2^{o(2^{2^{k-1} + k})}$.
\end{theorem}

\begin{proof}
	\Cref{lem:mA-good-mon-for-gap-det} informs us that there is a parallel monitor $\mV$ of length $O(k^2)$ that  
	recognizes $L_A^k$. Therefore, it also recognizes $L_A^k \cdot \Sigma_\$^\omega $ by \Cref{lem:ver-persistence}.
	\Cref{prop:monitor2automaton} tells us that $\mV$ can be turned into an alternating automaton with $l(\mV) = O(k^2)$ states that recognizes $L_A^k$.
	\Cref{cor:lower-for-det} yields that there is no deterministic monitor of length 
	$2^{o(2^{2^{k-1} + k})}$ that recognizes that language.
	From \cite{AceAFI:2017:CIAA}, we know that 
	if there were a DFA with $2^{o(2^k)}$ states that recognizes $L_A^k$, then there would be 
	a deterministic monitor of length $2^{2^{o(2^k)}}$ that recognizes $L_A^k$, which, as we  argued, cannot exist.  
	\qed 
\end{proof}

\paragraph{Hardness for regular monitors}
Proposition \ref{prp:lower_deterministic} does not guarantee that the $2^{O(n\cdot 2^n)}$ upper bound for the transformation from a parallel monitor to a nondeterministic regular monitor is tight.
To prove a tighter lower bound, let $L_U^k$ be the language that includes all strings of the form
$\# w_1 \# w_2 \# \cdots \# w_n \$ w $ 
where for $i = 1,\ldots,n$, $w_i \in W$, and $w \in \{ 0,1,\#,\$ \}^*$, and for every $i < n$, $w_i$ encodes a string that is smaller than the string encoded by $w_{i+1}$, in the lexicographic order.

\begin{lemma}\label{lem:lang2-fin-inf-trc}
	$\ftV \in L_U^k$
	if and only if
	$\forall \tV.~\ftV\tV
	\in L_U^k \cdot \Sigma_\$^\omega$.
	\qedmaybe 
\end{lemma}

We describe how $L_U^k$ can be recognized by a parallel monitor of size $O(k^2)$.
The idea is that we need to compare the encodings of two consecutive blocks of $l+1$ bits. Furthermore, a string is smaller than another if there is a position in these strings, where the first string has the value $0$ and the second $1$, and for every position that comes before that position, the bits of the two strings are the same.
We define the following monitors:


\begin{align*}
	smaller &= \sum_{\ftV \in \{ 0,1 \}^l} \left(
	\begin{array}{>{\displaystyle}c}
	find (\ftV 0) ~\paralC~ next_\# (find(\ftV 1))
	\\ ~\paralC~ \\ 
	\bigparalC_{\ftVV < \ftV} ~ \sum_{b \in \{0,1\}} (find(\ftVV b) ~\paralC~ next_\# (find(\ftVV b))) 
	\end{array}
	\right) 
	\\
	last & = \rec x(0.x + 1.x + \$.\yes)
	\\
	\mV_U &= \rec x(next_\#( perm ~\paralC~ (last ~\paralD~ (smaller ~\paralC~ x)) ))
\end{align*}




\begin{proposition}\label{prp:regular-bad}
	$\mV_U$ recognizes $L_U^k$ (and $L_U^k \cdot \Sigma_\$^\omega$) and $|\mV_U| = O(k^2)$.
	Every 
	regular monitor that recognizes $L_U^k$ or $L_U^k \cdot \Sigma_\$^\omega$ must be of length $2^{\Omega(2^k)}$.
	\qedmaybe
\end{proposition}

\section{Logical Consequences}
\label{sec:logical-consequences}

We now turn our attention back from the two monitoring systems to the corresponding logical fragments. We observe that the bounds that we have proved in the previous sections also apply when we discuss formula translations.
A version of \Cref{thm:samelogics} was proven in \cite{AcetoAFIL19}, but without complexity bounds.

\begin{theorem}\label{thm:samelogics}
	For every $\hV \in \ltmuC$ (\resp $\hV \in \ltmuS$), there is some $\hVV \in \CHML$ (\resp $\hVV \in \SHML$), such that $l(\hVV) = 2^{O(l(\mV) \cdot 2^{l(\mV)})}$ and  $\hSemL{\hV}=\hSemL{\hVV}$.
\end{theorem}

\begin{proof}
	We prove the case for $\hV \in \ltmuC$, as the case for $\hV \in \ltmuS$ is similar.
	By \Cref{thm:monitorability-maximality}, we know that there is a reactive parallel monitor $\mV$, such that $L_a(\mV) \cdot \Act^\omega = \hSemL{\hV}$ and $l(\mV) = O(l(\hV))$.
	By \Cref{prop:extended-mon-to-reg-mon}, we know that there is a regular monitor $\mVV$, such that $L_a(\mVV) = L_a(\mV)$ and $l(n) = 2^{O(l(\mV) \cdot 2^{l(\mV)})}$. We can then see that $l(\mVV) = 2^{O(l(\hV) \cdot 2^{l(\hV)})}$.
	According to \Cref{thm:monitorability-maximality}, there is a formula $\hVV \in \CHML$, such that $\hSemL{\hVV} = L_a(\mVV) \cdot \Act^\omega = L_a(\mV) \cdot \Act^\omega = \hSemL{\hV}$, and 
	$l(\hVV) = O(l(\mVV))$, yielding that $l(\hVV) = 2^{O(l(\hV) \cdot 2^{l(\hV)})}$.
	\qed 
\end{proof}

The cost of the construction in the proof of \Cref{thm:samelogics} is due to the regularization of the monitor. 
Our lower bounds --- and specifically 
\Cref{prp:regular-bad} 
--- demonstrate that this construction is optimal, because a better construction of $\hVV$ from $\hV$ would lower the cost of regularization via the synthesis functions.

\begin{theorem}\label{thm:logic-lower-bounds}
	There is some $\hV \in \maxHML$, such that for every $\hVV \in \SHML$, if $\hSemL{\hV}=\hSemL{\hVV}$, then $l(\hVV) = 2^{\Omega\left(2^{\sqrt{l(\hV)}}\right)}$.
\end{theorem}

\begin{proof}[Sketch]
	Otherwise, we could regularize 
	$\mV_U$ from \Cref{sec:lower} more efficiently than \Cref{prp:regular-bad} allows, by first turning $\mV_U$ to $\hV \in \maxHML$, then to $\hVV \in \SHML$, and finally to a regular monitor $\mV$.
The full proof is in the appendix.
	\qed 
\end{proof}

\begin{remark}
	We observe that 
	to prove \Cref{thm:logic-lower-bounds}, it was necessary to prove \Cref{prp:regular-bad} for regular monitors that are \emph{$\omega$-verdict equivalent}, and not just verdict equivalent, to $\mV_U$.
	The reason is that, in the proof of \Cref{thm:logic-lower-bounds}, the monitor $\mV$ that monitors for $\hVV$ is $\omega$-verdict equivalent to $\mV_U$ and there is no guarantee that it is, in fact, verdict equivalent to $\mV_U$.
	\qedd
\end{remark}

\begin{remark}
	In \cite{determinization}, the authors define a deterministic fragment of \SHML, which they then show to be equivalent to the full \SHML. 
	We can claim analogous bounds for translating formulae into this smaller fragment, using similar arguments to those used above. We omit a full exposition of this claim.
	\qedd 
\end{remark}

\section{Conclusion}
\label{sec:conclusions}

We determined the cost of turning a parallel monitor into an equivalent regular, or deterministic, monitor.
As a result, we saw that, over infinite traces, \maxHML is doubly-exponentially more succinct than \SHML.

Regular monitors were introduced in \cite{FraAI:17:FMSD} to monitor for \SHML over processes.
 The cost of determinization of regular monitors was examined in \cite{AceAFI:2017:CIAA,determinization}.
 Aceto \etal in  \cite{aceto_et_al:LIPIcs:2018:9572} used a similar determinization process on formulae in the context of enforcement. 


In \cite{AcetoAFIL19}, we also synthesized \emph{tight} monitors, which are monitors that reach a verdict as soon as they have analysed enough information from the trace, and not later.
It is often important to reach a verdict as soon as possible, but it is also important to avoid  burdening a monitored system with a very large monitor.
Therefore, 
it would also be of interest to  determine  how much it costs to turn a parallel or regular monitor into a verdict-equivalent tight monitor.
This is a topic that we leave for future work.
%


\bibliographystyle{plain}
\bibliography{refs2}

\newpage
\appendix
\section*{Appendix}

This appendix gathers the omitted proofs from the main text of this paper.

\subsection*{Omitted Proofs from \Cref{sec:transformations}}

  \begin{proof}[of \Cref{cor:parallel2NFA}]
	The alternating automaton that is constructed in the proof of \Cref{prop:monitor2automaton} has at most as many states as there are submonitors in $\mV$ which, in turn, are no more than $l(\mV)$.
	To conclude, every alternating automaton with $k$ states can be converted into an NFA with at most $2^k$ states that recognises the same language~\cite{alternation,FJY90}.
	\qed  \end{proof}

The following technical lemmata help in proving the improved upper bound in \Cref{prop:better_determinization}.

\begin{lemma}\label{lem:one-trace-twostuff}
	Let $\mV$ be a regular monitor and $\mVV$ a submonitor of $\mV$, such that $\mVV$ is a sum of $\acta.\yes$ and of $\actb.\yes$. Then, there is some $\ftr \in \Act^*$, such that $\mV$ accepts both $\ftr \acta$ and $\ftr \actb$.
\end{lemma}
\begin{proof}
	By straightforward induction on the construction of $\mV$ from $\mVV$.
\qed  \end{proof}

\begin{remark}
	Monitor $\mV_\actt$ from \Cref{rem:infinitestate} demonstrates that \Cref{lem:one-trace-twostuff} does not hold for parallel monitors, and so does
	$\no \paralC ( a.\yes + b.\yes )$. 
\qedd 
\end{remark}

\begin{lemma}\label{lem:immediate-yes}
	Let $\mV$ be a 
	closed regular monitor that does not have any submonitor of the form $\rec x \yes$, and let $\mV'$ be a submonitor of $\mV$.
	If $\mV' \wtraS{} \yes$, then $\mV' = \yes$.
\end{lemma}

\begin{proof}
	If $\mV' \wtraS{} \yes$, then $\mV' \traS{\tau}^k \yes$ for some $k \geq 0$.
	We prove that if $k>0$, then $\rec x \yes$ is a submonitor of $\mV$.
	If $k=1$, then $\mV'$ is a sum of $\rec x \yes$, 
	and so $\rec x \yes$ is a submonitor of $\mV'$, and therefore of $\mV$.
	If $\mV' \traS{\tau} \mVV$ and  $\mVV \traS{\tau}^{k-1} \yes$, then by \Cref{lem:subs2subs},
	$\mVV$ is a submonitor of $\mV$, so
	$\rec x \yes$ is a submonitor of $\mVV$ by induction, and therefore of $\mV$.
\qed  \end{proof}

\begin{proof}[of \Cref{prop:better_determinization}]
	We recall that, in the proof of  \Cref{prop:determinization} given in \cite{AcetoAFIL19}, the construction starts from a
	consistent reactive parallel monitor $\mV$ and turns it into two NFAs $A_a$ and $A_r$, such that $L(A_a) = L_a(\mV)$, $L(A_r) = L_r(\mV)$, and the number of states in $A_a$ and in $A_r$ is at most $2^{O(l(\mV))}$ (\Cref{cor:parallel2NFA}).
	Then, these automata can be turned into regular monitors $\mV_a$ and $\mV_r$, and we can conclude by determinizing $\mV^a_R + \mV^r_R$. 
	To improve the upper bound for the construction, we determinize starting from the NFAs.
	Let $A_a = (Q_a,\Act,\delta_a,q_0^a,F_a)$ and $A_r = (Q_r,\Act,\delta_r,q_0^r,F_r)$.
	
	We first construct a new NFA\footnote{For clarity, the construction of the NFA $N$ makes a mild use of $\varepsilon$-transitions.}, $N = (Q,\Act \cup \{ \yesac,\noac, \varepsilon \},\delta,q_0,F)$, such that:
	$Q = \{q_0,f\} \dcup Q_a \dcup Q_r $, where $\dcup$ stands for disjoint union; $F = \{f\}$; and 
	\begin{align*}
	\delta(q,\act) =  \begin{cases}
	\{ q_0^a, q_0^r \} & \text{for } q = q_0,~ \act = \varepsilon, \\
	\delta_o(q,\act)  & \text{for } q \in Q_o,~ \act \in \Act \cup \{\varepsilon \},~ o \in \{a,r\},  \\
	\{ f \}  & \text{for } q \in F_a,~ \act = \yesac,  \\
	\{ f \}  & \text{for } q \in F_r,~ \act = \noac,  \\
	\{ f \}  & \text{for } q = f,  \\
	\emptyset & \text{otherwise.}
	\end{cases}
	\end{align*}
	It is not hard to see that the set of minimal traces in $L(N)$, with respect to the prefix order, is
	$\min L(N) = L(A_a)\cdot \{\yesac\} ~\cup~ L(A_r)\cdot \{\noac\}$.
	
	From \cite[Theorem 2]{AceAFI:2017:CIAA}, 
	there is a deterministic (closed) monitor $\mVV'$, such that $L_a(\mVV') = L(N)$ and $l(\mVV') = 2^{2^{O(|Q|)}} = 2^{2^{2^{O(l(\mV))}}}$.
	We observe that there cannot be a submonitor of $\mVV'$ that is a sum both of $\yesac.\yes$ and of $\noac.\yes$ --- otherwise, by \Cref{lem:one-trace-twostuff}, $L_a(\mV)\cap L_r(\mV) \neq \emptyset$, which is a contradiction, because $\mV$ is consistent.
	We further assume that $\mVV'$ has no submonitors of the form $\rec x \yes$, $\rec x \no$, as these can be replaced by $\yes$ and $\no$, respectively, yielding an equivalent monitor.
	Let $\mVV$ be the result of replacing in $\mVV'$ all maximal sums of $\yesac.\yes$ by $\yes$ and of $\noac.\yes$ by $\no$.
	In particular,
	there are monitor variables $x_1,\ldots,x_k,y_1,\ldots,y_l$ that do not appear in $\mVV'$ and an open monitor $\mVV_o(x_1,\ldots,x_k,y_1,\ldots,y_l)$, such that $\mVV'=\mVV_o(s_1,\ldots,s_k,s'_1,\ldots,s'_l)$, where $s_1,\ldots,s_k$ are (all the occurrences of) sums of $\yesac.\yes$ in $\mVV'$, $s'_1,\ldots,s'_l$ are (all the occurrences of) sums of $\noac.\yes$ in $\mVV'$, and $\mVV = \mVV_o(\yes,\ldots,\yes)$.
	We prove that 
	$\min L_a(\mVV') = L_a(\mVV)\cdot \{\yesac\} ~\cup~ L_r(\mVV)\cdot \{\noac\}$.
	
	
	Let $w \in \min L_a(\mVV')$. Because $L_a(\mVV') = L(N)$, we know that $w \in L(N)$, and therefore $w \in L(A_a)\cdot \{\yesac\} ~\cup~ L(A_r)\cdot \{\noac\}$.
	So,
	it is either the case that $w = w' \yesac$ for $w' \in L_a(A_a)$, or that $w = w' \noac$  for $w' \in L_a(A_a)$ --- w.l.o.g. we assume the first case.
	Therefore, $\mVV' \wtraS{w'\yesac} \yes$ and since $w \in \min L_a(\mVV')$, $\mVV' \centernot{\wtraS{w'}}\yes$.
	%
	We demonstrate that $\mVV \wtraS{w'} \yes$.
	Specifically, 
	%
	\emph{we prove by induction on an arbitrary submonitor $\mVV'_o = \mVV'_o(x_1,\ldots,x_k,y_1,\ldots,y_l)$ of $\mVV_o(x_1,\ldots,x_k,y_1,\ldots,y_l)$ that if $\mVV_s=\mVV'_o(s_1,\ldots,s_k,s'_1,\ldots,s'_l) \wtraS{w'\yesac} \yes$ and $
		\mVV_s
		\centernot{\wtraS{w'}} \yes$ for some $w' \in \Act^*$, then $\mVV_y = \mVV'_o(\yes,\ldots,\yes) \wtraS{w'} \yes$.}
	\begin{description}
		\item[The base cases are:] $\mVV'_o$ is a verdict, 
		which is immediate, because the claim's assumptions cannot both hold; or $\mVV'_o = x_i$, in which case $\mVV_y = \yes$. 
		\item[If for some $\act \in \Act\cup\{\yesac,\noac \}$, $\mVV'_o = \act.\mVV''$] and 
		$\mVV_s \wtraS{w'\yesac} \yes$ and $\mVV_s \centernot{\wtraS{w'}} \yes$, then there are two cases. 
		In the first case, $\act = \yesac$. 
		Then, 
		as $w' \in \Act^*$, we have that $w' = \varepsilon$, so $\mVV''(s_1,\ldots,s_k,s'_1,\ldots,s'_l) \wtraS{} \yes$, and by \Cref{lem:immediate-yes} ($\mVV_s$ is a submonitor of $\mVV'$),
		$\mVV''(s_1,\ldots,s_k,s'_1,\ldots,s'_l) = \yes$,
		and therefore 
		$\mVV'_o = \yesac.\yes$, which contradicts our definitions.
		Otherwise, for $\act \neq \yesac$,
		$w' = \act~w''$, and $\mVV''(s_1,\ldots,s_k,s'_1,\ldots,s'_l) \wtraS{w''\yesac} \yes$ and $\mVV''(s_1,\ldots,s_k,s'_1,\ldots,s'_l) \centernot{\wtraS{w''}} \yes$. 
		By the inductive hypothesis,  $\mVV''(\yes,\ldots , \yes) \wtraS{w''} \yes$, and therefore $\mVV_y \wtraS{w'} \yes$.
		\item[If $\mVV'_o  = \mVV_1+\mVV_2$] and $\mVV_s \wtraS{w'\yesac} \yes$ and $\mVV_s 
		\centernot{\wtraS{w'}} \yes$, then 
		we know that $\mVV_s$ is not a sum of $\yesac.\yes$.
		Therefore, w.l.o.g. 
		$\mVV_1(s_1,\ldots,s_k,s'_1,\ldots,s'_l) \wtraS{w'\yesac} \yes$ and $\mVV_1(s_1,\ldots,s_k,s'_1,\ldots,s'_l) \centernot{\wtraS{w'}} \yes$,
		and we are done by the inductive hypothesis.
		%
		%
		%
		\item[The case for $\mVV'_o = \rec x \mVV''$] is more straightforward, as $\mVV'_o$ and $\mVV''$ have the same weak transitions.
	\end{description}
	
	We now prove that if $w\in L_a(\mVV)\cdot \{\yesac\} ~\cup~ L_r(\mVV)\cdot \{\noac\}$, then $w \in \min L_a(\mVV')$.
	Let $w \in L_a(\mVV)\cdot \{\yesac\}$ (the case for $w \in L_r(\mVV)\cdot \{\noac\}$ is symmetric).
	Then there is a $w' \in L_a(\mVV)$, such that $w = w'.\yesac$.
	Let $w''$ be the shortest prefix of $w'$, such that $w'' \in L_a(\mVV)$.
	We observe that if $w''.\yesac \in \min L_a(\mVV')$, then $w''.\yesac \in \min L(N)$, and therefore $w'' \in L(A_a)$, yielding that $w' \in L(A_a)$, because $A_a$ is suffix-closed (by verdict-persistence, \Cref{lem:ver-persistence}), and therefore $w = w'.\yesac \in \min L_a(\mVV')$.
	Furthermore, if $w''.\yesac \in L_a(\mVV')$, we can see that $w''.\yesac \in \min L_a(\mVV')$, as $w'' \in \Act^*$.
	Thus, it suffices to prove that  $w''.\yesac \in L_a(\mVV')$.
	%
	By induction on the derivation, and due to the minimality of $w''$, we can see that either $\mVV_o(x_1,\ldots,x_k,y_1,\ldots,y_l) \wtraS{w''} \yes$, in which case $\mVV' \wtraS{w''} \yes$, and therefore $\mVV' \wtraS{w''.\yesac} \yes$, or 
	$\mVV_o(x_1,\ldots,x_k,y_1,\ldots,y_l) \wtraS{w''} x_i$ for some $x_i$, in which case $\mVV' \wtraS{w''} s'_i$, where $s'_i$ is $s_i$ after applying a number of steps substituting variables for monitors. As $s_i$ is a sum of $\yesac.\yes$, so is $s'_i$, therefore $s'_i \traS{\yesac} \yes$.
	
	We just demonstrated that $\min L_a(\mVV') = L_a(\mVV)\cdot \{\yesac\} ~\cup~ L_r(\mVV)\cdot \{\noac\}$. But we have also seen that $L_a(\mVV') = L(N)$ and that $\min L(N) = L(A_a)\cdot \{\yesac\} ~\cup~ L(A_r)\cdot \{\noac\}$, and therefore $L_a(\mVV) = L(A_a) = L_a(\mV)$ and $L_r(\mVV) = L(A_r) = L_r(\mV)$, which is what we wanted to show.
	\qed  \end{proof}

\subsection*{Omitted Proofs from \Cref{sec:lower}}

\begin{proof}[of \Cref{lem:lang1-fin-inf-trc}]
	The ``only if'' direction is immediate, and therefore we only show the ``if'' direction. Let $\ftV$ be such that $\forall \tV.~\ftV\tV
	\in L_A^k \cdot \Sigma_\$^\omega$. 
	From our assumptions about $\ftV$, $\ftV\#^\omega
	\in L_A^k \cdot \Sigma_\$^\omega$. According to the definition of $L_A^k$, there must be a finite prefix $\ftVV$ of $\ftV\#^\omega$, such that $\ftVV \in L_A^k$ and ends with $\$ $. Since the only symbol in $\#^\omega$ is $\#$, $\ftVV$ is a prefix of $\ftV$.
	Therefore, since $L_A^k$ is suffix-closed, we conclude that $\ftV \in L_A^k$.
	\qed 
\end{proof}

\begin{proof}[of \Cref{lem:skipsharp}]
	The ``if'' direction is straightforward, using \Cref{cor:parallel-and-monitors}. For the ``only if'' direction,
	notice that the $\#.\mV$ component of the monitor is the only component that can produce a $\yes$ verdict, and it activates at the occurrences of the $\#$ symbol.
	Therefore, if $skip_\#(\mV)$ accepts $\fftV$, there must be some 
	$\ftV \# \fftVV = \fftV$, such that $\mV$ accepts $\fftVV$.
	Let $\ftV$ be minimal for this to happen.
	
	We now prove by contradiction that $\ftV \in \{0,1,\#\}^*$. 
	If that is not the case, $\$ $ appears in $\ftV$, so there are $\ftV_1\$\ftV_2 = \ftV$.
	When reading $\ftV_1\$$, all monitor components of $skip_\#(\mV)$ must fail (reach \stp), except the ones that have come from $\#.\mV$.
	Therefore, we can split $\fftV$ as $\ftV_1 \# \fftVV' = \fftV$, such that $\mV$ accepts $\fftVV'$, which is a contradiction, due to the assumed minimality of $\ftV$.
	Therefore, $\ftV \in \{0,1,\#\}^*$.
	\qed  \end{proof}

\paragraph{Hardness of $L_A^k$ for deterministic monitors}
We introduce the notion of a simple trace. As \Cref{lem:simple_traces} reveals, the cardinality of a set of simple traces gives a lower bound on the size of a regular monitor.


\begin{definition}
	We call a derivation $ \mV \wtraS{\ftV} \mV'$ simple, if rules \rtit{RecB} and \rtit{Ver} are not used in the proof of any transition of the derivation.
	We say that a trace $\ftV \in \Sigma_\$^*$ is simple for monitor $\mV$ if there is a simple  derivation $ \mV \wtraS{\ftV} \mV'$. We say that a set of simple traces for $\mV$ is simple for $\mV$.
	\exqed
\end{definition}

\begin{lemma}[Lemma 20 from \cite{determinization}]\label{lem:simple_traces}
	Let $\mV$ be a regular monitor and $G$ a (finite) simple set of traces for $\mV$. Then, $l(\mV) \geq |G|$.
	\qed 
\end{lemma}

The following \Cref{lem:pump_not_simple,lem:find_px_det} are variations of the Pumping Lemma for monitors.

\begin{lemma}[Lemma 22 from \cite{determinization}]\label{lem:pump_not_simple}
	Let $\mV \wtraS{\ftV} \mVV$, such that $\mV$ is regular and $\ftV$ is not simple for $\mV$. Then, there are $\ftV = \ftV_1\ftV_2\ftV_3$, such that $|\ftV_2| >0$ and for every $i \geq 0$,  $\mV \wtraS{\ftV_1\ftV_2^i \ftV_3} \mVV$.
	\qed 
\end{lemma}

\begin{lemma}[Lemma 27 from \cite{determinization}]\label{lem:find_px_det}
	Let $\mV \wtraS{\ftV} \mV'$, such that $\ftV$ is not simple for $\mV$ and $\mV$ is deterministic (and regular). Then, there are $\ftV = \ftV_1\ftV_2\ftV_3$ and monitor $\mVV$, such that $|\ftV_2| >0$ and  $\mV \wtraS{\ftV_1} \mVV \wtraS{\ftV_2} \mVV \wtraS{\ftV_3} \mV'$.
\end{lemma}

Fix a partition $\{C,D\}$ of $\{0,1\}^k$, such that $|C|=|D| = 2^{k-1}$; let $K = 2^{2^{k-1}}$.
Let $P_C = C_0\cdots C_{K-2}$ be a permutation of $2^C \setminus \{\emptyset\}$, and $P_D = D_0\cdots D_{K-2}$ a permutation of $2^D \setminus \{\emptyset\}$.
For every $0\leq i<K-1$, let
	$$t_i = \# enc(w_0) \# enc(w_1) \# \cdots \# enc(w_{|C_i|}) \#,$$ where $w_0 w_1 \cdots w_{|C_i|}$ is a permutation of $C_i$
	and
	$$s_i = \# enc(w'_0) \# enc(w'_1) \# \cdots \# enc(w'_{|D_i|}) \#,$$ where $w'_0 w'_1 \cdots w'_{|D_i|}$ is a permutation of $D_i$.
	Thus, each $t_i$ and $s_i$ encodes a permutation
	of a distinct subset of $\{0,1\}^k$.
Let $$ t_K(P_C,P_D) = t_0 \$ s_0 \$ t_1 \$ s_1 \$ \cdots \$t_{K-1} \$ s_{K-1} $$
and let
	\[T = \{ t_K(P_C,P_D) \mid P_C \in (2^C \setminus \{\emptyset\})! \text{ and } P_D \in (2^D \setminus \{\emptyset\})! \}.\]
	Notice that  $|T| = ((K-1)!)^2$.

The following \Cref{lem:tK_and_prefixes} 
characterizes trace $t_K$, as defined in \Cref{sec:lower}, and is useful in the proof of \Cref{cor:tK_and_prefixes}, which follows immediately from it.

\begin{lemma}\label{lem:tK_and_prefixes}
	If $f\neq g$ are prefixes of $t_K$, then there is some $h \in \{0,1,\#,\$\}^*$, such that $fh \in L_A^k$ iff $gh \notin L_A^k$.
\end{lemma}
\begin{proof}
	We can assume that $f$ is a proper prefix of $g$. We have the following cases:
	\begin{description}
		\item[The symbol $\$ $ appears more times in $g$ than in $f$.]
		Then, for some $f',g' \in \Sigma^*$ and $0 \leq i \leq j$:
		\begin{itemize}
			\item
			$f = f'$ and ($g =  t_0 \$ s_0 \$ \cdots \$ s_j \$ g' $ or $g =  t_0 \$ s_0 \$ \cdots \$ t_j \$ g' $); or
			\item
			$f =  t_0 \$ s_0 \$ \cdots \$ t_i \$ f' $
			and
			$g =  t_0 \$ s_0 \$ \cdots \$ s_j \$ g' $;
			or
			\item
			$f = t_0 \$ s_0 \$ \cdots \$ s_{i} \$ f' $
			and
			$g =  t_0 \$ s_0 \$ \cdots \$ t_j \$ g' $;
			or
			\item
			$f =  t_0 \$ s_0 \$ \cdots \$ t_i \$ f' $
			and
			$g =  t_0 \$ s_0 \$ \cdots \$ t_j \$ g' $ and $i<j$;
			or
			\item
			$f =  t_0 \$ s_0 \$ \cdots \$ s_i \$ f' $
			and
			$g =  t_0 \$ s_0 \$ \cdots \$ s_j \$ g' $ and $i<j$.
		\end{itemize}
		Then, for all these cases we can immediately see that there is some $w \in W$ that appears in one of $t_i,s_i,t_j, s_j$, such that for $h = \# w \# \$ $, $fh \in L_A^k$ iff $gh \notin L_A^k$.
		
		\item[The symbol $\# $ appears more times in $g$ than in $f$, but $\$ $ does not.]
		Then, for some $f',g' \in \{0,1\}^*$,
		\begin{align*}
		f &=  t' \$ \# enc(w_0) \# enc(w_1) \# \cdots \# enc(w_{i}) \# f' 
		~~~\text{ and}\\
		g &=  t' \$ \# enc(w_0) \# enc(w_1) \# \cdots \# enc(w_{j}) \# g' 
		\end{align*}
		Then, for $h = \$ \# enc(w_{j}) \# \$ $, $gh \in L_A^k$ and $fh \notin L_A^k$.
		
		\item[The symbols $\$,\# $ appear the same number of times in $g$ and in $f$.]
		Then, for some $f',g' \in \{0,1\}^*$,
		\begin{align*}
		f &=  t' \$ \# enc(w_0) \# enc(w_1) \# \cdots \# enc(w_{i}) \# f' 
		~~\text{ and}\\
		g &=  t' \$ \# enc(w_0) \# enc(w_1) \# \cdots \# enc(w_{i}) \# f' g' 
		\end{align*} 
		and $|g'|>0$.
		Therefore, $f'g'$ is a prefix of $enc(w_{i+1})$. Let $h_0$ be such that $f'g'h_0 = enc(w_{i+1})$.
		Then, for $h = h_0 \#\$ \# enc(w_{i+1}) \# \$ $, $gh \in L_A^k$ and $fh \notin L_A^k$.
		\qedhere
		
	\end{description}
\end{proof}

\begin{cor}\label{cor:tK_and_prefixes}
	If $\ftV\neq \ftVV$ are prefixes of $t_K$, then there is some $\tV \in \Sigma_\$^\omega$, such that $\ftV\tV \in L_A^k \cdot \Sigma_\$^\omega$ iff $\ftVV\tV \notin L_A^k \cdot \Sigma_\$^\omega$.	
\end{cor}

\begin{proof}
	By \Cref{lem:tK_and_prefixes,lem:lang1-fin-inf-trc}.
	\qed  \end{proof}

\begin{proposition}\label{prop:tKsimple}
	If $\mV$ is a deterministic regular monitor that recognizes $L_A^k \cdot \Sigma_\$^\omega$, then $t_{K}$ is a simple trace for $\mV$.
\end{proposition}

\begin{proof}
	If $t_K$ is not simple, then by  \Cref{lem:find_px_det}, $t_K = \ftV_1\ftV_2\ftV_3$, such that $|\ftV_2|>0$ and there is a monitor $\mVV$, such that
	$\mV \wtraS{\ftV_1}\mVV \wtraS{\ftV_2}\mVV$. But according to \Cref{cor:tK_and_prefixes}, there is some $\tV$, such that $\ftV_1\tV \in L_A^k \cdot \Sigma_\$^\omega$ iff $\ftV_1\ftV_2\tV \notin L_A^k \cdot \Sigma_\$^\omega$.
	This
	is a contradiction, because $\ftV_1\tV \in L_A^k \cdot \Sigma_\$^\omega$ iff $\mVV$ accepts $\tV \cdot \Sigma_\$^\omega$ iff $\ftV_1\ftV_2\tV \in L_A^k \cdot \Sigma_\$^\omega$.
	\qed  \end{proof}

\begin{proof}[of \Cref{cor:lower-for-det}]
	An immediate consequence of \Cref{prop:tKsimple,lem:simple_traces}.
	\qed 
\end{proof}

\paragraph{The hardness of regularization}

\begin{proof}[of \Cref{prp:regular-bad}]
	First, we sketch
	the proof of the first part of the proposition.
	We observe that
	$smaller$ accepts
	a trace $w_1\# w_2 \tV$, where $w_1,w_2 \in W$,
	exactly when $w_1$ encodes a smaller sequence than $w_2$, with respect to the lexicographic ordering.
	Then, $\mV_U$ keeps reading blocks of bits between the $\# $ separators, while ensuring that each of these is an element of $W$ (using monitor $perm$), and that it either is the last such block of bits (using monitor $last$), or that it encodes a smaller sequence than the next one (using monitor $smaller$). 
	
	It now suffices to prove that, for any ordered sequence $a_1 a_2 \cdots a_{c}$ of strings from $\{0,1\}^k$, the finite trace
	$s = \# enc(a_1) \# enc(a_2) \# \cdots \# enc(a_{c}) $ is simple for any regular monitor that accepts exactly the infinite extensions of $L_U^k$.
	If $\ftV$ is not simple for $\mV$, then, according to Lemma \ref{lem:pump_not_simple},
	for every $\mV \wtraS{\ftV} \mVV$, there are
	$\ftV = \ftV_1\ftV_2\ftV_3$, such that $|\ftV_2|>0$ and  for every $i \geq 0$, $\mV \wtraS{\ftV_1 \ftV_2^i \ftV_3} \mVV$.
	Therefore, if $\mV \wtraS{\ftV } \mVV \wtraS{\$0^j} \yes $ for some $j \geq 0$, then
	$\mV \wtraS{\ftV_1 \ftV_2^2 \ftV_3} \mVV \wtraS{\$0^j} \yes $ and $\mV \wtraS{\ftV_1 \ftV_2^3 \ftV_3} \mVV \wtraS{\$0^j} \yes $.  
	If $|\ftV_2|$ is not a multiple of $k(l+1) + 1$, then
	it is not hard to see that we reach a contradiction, as $\ftV_1 \ftV_2^2 \ftV_3 \$ 0^j \notin L_U^k$%
	, because it includes a block of bits that is not in $W$.
	Otherwise, 
	$\ftV_2$ is of the form $f \# g$ or $f \# enc(a_i) \# \cdots \# enc(a_{i+j}) \# g$, where $fg \in \{0,1\}^{k(l+1)}$ and $j \geq 0$.
	If $gf \notin W$, then $\ftV_1 \ftV_2^2 \ftV_3 \$ 0^j \notin L_U^k$, and again, we have a contradiction.
	But even if $gf \in W$, we see that $gf$ appears twice in $\ftV_1 \ftV_2^3 \ftV_3$, and again we have a contradiction, as $\ftV_1 \ftV_2^3 \ftV_3 \$ 0^j \notin L_U^k$, because 
	there appear two elements of $W$ in $\ftV_1 \ftV_2^3 \ftV_3$ that encode the same string.
	%
	%
	\qed  \end{proof}


\subsection*{Omitted Proofs from \Cref{sec:logical-consequences}}

\begin{proof}[of \Cref{thm:logic-lower-bounds}]
	Let $k$ and $\mV_U$ be as defined in \Cref{sec:lower}.
	\Cref{thm:monitorability-maximality} asserts the existence of a formula $\hV$, such that $l(\hV) = O(l(\mV_U))$ and $\hSem{\hV} = L_a(\mV_U)\cdot \Sigma_\$^\omega = L_U^k\cdot \Sigma_\$^\omega$.
	Now, let $\hVV \in \SHML$, where $\hSemL{\hV}=\hSemL{\hVV}$.
	According to \Cref{thm:monitorability-maximality}, there is a regular monitor $\mV$, such that $l(\mV) = O(l(\hVV))$, and $\hSem{\hV} = \hSemL{\hVV} = L_a(\mV)\cdot \Sigma_\$^\omega$.
	Therefore, $\mV$ recognizes $L_U^k\cdot \Sigma_\$^\omega$, and
	by \Cref{prp:regular-bad}, $l(\mV) = 2^{\Omega(2^k)} = 2^{\Omega(2^{\sqrt{l(\hV)}})}$.
	This yields that $l(\hVV) = 2^{\Omega(2^k)} = 2^{\Omega(2^{\sqrt{l(\hV)}})}$.
	\qed 
\end{proof}

\end{document}